\documentclass[11pt]{article}
\usepackage{amssymb,amsmath,amsthm,mathrsfs}
\usepackage[dvips]{graphicx}

\newtheorem{theorem}{Theorem}

\newtheorem{proposition}{Proposition}

\newtheorem{lemma}{Lemma}
\newtheorem{definition}{Definition}

\renewcommand{\H}{{\cal H}}
\newcommand{\tw}{{\mathbf{tw}}}

\newcommand{\ifb}{{\mathbf{fibn}}}
\newcommand{\ifl}{{\mathbf{ifl}}}

\newcommand{\val}{{\mathbf{val}}}

\newcommand{\fbn}{{\mathbf{fbn}}}
\newcommand{\aw}{{\mathbf{aw}}}
\newcommand{\fhw}{{\mathbf{fhw}}}

\newcommand{\ghw}{{\mathbf{ghw}}}

\newcommand{\mw}{{\mathbf{mw}}}

\usepackage{a4wide}

\begin{document}

\author{
Fedor V. Fomin\thanks{%
Department of Informatics, University of Bergen, N-5020 Bergen,
Norway} \addtocounter{footnote}{-1}
\and Petr A. Golovach\footnotemark
\and Dimitrios M. Thilikos\thanks{Department of Mathematics,
        University of Athens, Panepistimioupolis, GR15784 Athens,
        Greece}
}

\title{Approximating Acyclicity Parameters\\  of Sparse Hypergraphs}
\date{}
\maketitle
\begin{abstract}
\noindent The notions of hypertree width and generalized hypertree
width were introduced by  Gottlob, Leone, and Scarcello (PODS'99, PODS'01) in order to
extend the concept of hypergraph acyclicity. These notions were further generalized by Grohe and Marx in SODA'06, who introduced the fractional hypertree width of a hypergraph.  All these width parameters on hypergraphs are  useful for
extending tractability of many  problems in database theory and
artificial intelligence. Computing each of these width parameters is
known to be an {\sf NP}-hard problem.
Moreover, the (generalized) hypertree width of an $n$-vertex hypergraph  cannot be approximated
within a factor $c\log n$ for some   constant $c>0$ unless ${\sf
P}\neq {\sf NP}$.
In this paper, we study the approximability of
(generalized, fractional) hyper treewidth of sparse hypergraphs where the
criterion of sparsity reflects the sparsity of   their incidence
graphs. Our first step is to prove that the (generalized, fractional) hypertree
width of a hypergraph ${\cal H}$ is constant-factor sandwiched  by
the treewidth of its incidence graph, when the incidence graph
belongs to some apex-minor-free graph class (the family of
apex-minor-free graph classes  includes planar graphs and graphs of
bounded genus). This determines the combinatorial borderline above which the notion of
(generalized, fractional) hypertree width becomes essentially more general than
treewidth, justifying that way its functionality as a hypergraph acyclicity measure.
While for more general sparse families of hypergraphs
treewidth of incidence graphs and all  hypertree width parameters may differ arbitrarily,
there are sparse families where a constant factor approximation
algorithm is possible. In particular,  we give a constant factor
approximation polynomial time algorithm for (generalized, fractional) hypertree
width on hypergraphs whose incidence graphs belong to some $H$-minor-free graph class.
This extends the results of Feige, Hajiaghayi,
 and Lee from STOC'05 on approximating  treewidth of $H$-minor-free graphs.
\end{abstract}

\section{Introduction}
Many important  theoretical and ``real-world'' problems can be
expressed as constrained satisfaction problems (CSP).
Among   examples one can mention numerous problems from different domains like
Boolean satisfiability, temporal
reasoning, graph coloring, belief maintenance, machine vision, and scheduling. Another example is the
conjunctive-query containment problem, which is  a fundamental problem in database query evaluation. In fact, as it was shown by
Kolaitis and Vardi \cite{KolaitisV00}, CSP, conjunctive-query containment, and  finding homomorphism for relational structures  are essentially the same problem.
The problem  is known to be
{\sf NP}-hard in general \cite{ChandraM77}  and polynomial time solvable for
restricted class of acyclic queries \cite{Yannakakis81}.
Recently, in the database and  constraint satisfaction communities various extensions of query (or hypergraph) acyclicity were studied.
 The main motivation for the quest for  a suitable
 measure of acyclicity of a hypergraph (query, or relational structure) is the extension of  polynomial time solvable cases (like acyclic hypergraph) to more general instances.
 In this direction, Chekuri and Rajaraman in \cite{ChekuriR00} introduced the notion of query width.
Gottlob,  Leone, and Scarcello \cite{GottlobLS00,GottlobLS01,GottlobLS03}
defined hypertree width and generalized hypertree width. Furthermore, Grohe and Marx \cite{GroheM06} have introduced the most general parameter known so far, fractional hypertree width, and proved that
CSP, restricted to instances of bounded fractional hypertree width, is polynomial time solvable.

Unfortunately,  all known variants of hypertree width are {\sf NP}-complete \cite{GottlobGMSS05,GottlobMS07}.
Moreover, generalized hypertree width  is {\sf NP}-complete
even when checking whether its value is at most 3 (see ~\cite{GottlobMS07}).
 In the case of hypertree width,
the problem is $W[2]$-hard when parameterized by $k$   \cite{GottlobGMSS05}.
Both hypertree width  and the generalized hypertree are hard to approximate.
For example, the reduction  of Gottlob et al. in~\cite{GottlobGMSS05} can be used to show that
the generalized
hypertree width of an $n$-vertex hypergraph  cannot be approximated
within a factor $c\log n$ for some   constant $c>0$
unless ${\sf P}\neq {\sf NP}$.

All these parameters for hypergraphs can be seen as  generalizations of the
treewidth of a graph.
The treewidth is a fundamental graph parameter from Graph Minors
Theory by Robertson and Seymour \cite{RobertsonS86} and it has
numerous algorithmic applications (for a  survey, see~\cite{Bodlaender98}).
It is an old open question whether the treewidth
can be approximated within a constant factor and the best known approximation
algorithm for treewidth is $ \sqrt{\log OPT}$-approximation due to Feige et al. \cite{FeigeHajLee05}.  However, as it was shown by
 Feige et al. \cite{FeigeHajLee05}, the treewidth of an $H$-minor-free graph is constant factor approximable.

\medskip
\noindent\textbf{Our results.} Our first result is combinatorial.
We show that for a wide family of hypergraphs (those where the incidence graph
excludes an apex graph as a minor -- that is a graph that can become planar after removing a vertex) the fractional and generalized hypertree width of a hypergraph is
bounded by a linear function of treewidth of its incidence graph. Apex-minor-free graph classes
include planar and bounded genus graphs.

For hypergraphs whose incidence graphs are apex graphs
the two parameters may differ arbitrarily, and this result determines
the boundary where fractional hypertree width starts being essentially different from treewidth of the incidence graph.
This indicates that hypertree width parameters are more useful as the adequate version of acyclicity for non-sparse  instances.

Our proof is based on theorems from bidimensionality
theory and a min-max (in terms of fractional hyperbrambles) 
characterization of fractional hypertree width.
The proof essentially identifies what is the obstruction analogue of fractional hypertree width
for incidence graphs.

Our second result applies further for sparse classes where the difference between
(generalized, fractional) hypertree width of a hypergraph and treewidth of its incidence graph
can be arbitrarily large. In particular, we give a constant factor
approximation algorithm for generalized and fractional hypertree width of hypergraphs with $H$-minor-free incidence graphs extending the results of Feige et al. \cite{FeigeHajLee05}
from treewidth to (generalized, fractional) hypertree width.
The algorithm is based on a series of theorems based on the main decomposition theorem of the Robertson-Seymour's Graph Minor
project. As a combinatorial corollary of our results, it follows that generalized hypertree width and
fractional hypertree width differ  within constant multiplicative factor if the incidence graph of the hypergraph  does not contain a fixed graph as a minor.

\section{Definitions and preliminaries}

\subsection{Basic definitions}

We consider finite undirected graphs without loops or multiple
edges. The vertex set of a graph $G$ is denoted by $V(G)$ and its
edge set by $E(G)$ (or simply by $V$ and $E$ if it does not create
confusion).

Let $G$ be a graph. For a vertex $v$, we denote by $N_G(v)$ its
\emph{(open) neighborhood}, i.e. the set of vertices
which are adjacent to $v$. The \emph{closed neighborhood} of $v$,
i.e. the set $N_G(v)\cup\{v\}$, is denoted by $N_G[v]$. For
$U\subseteq V(G)$, we define  $N_G[U]=\bigcup_{v\in U}N_G[v]$ (we may omit
index if the graph under consideration is clear from the context).
If $U\subseteq V(G)$ (or $u\in V(G)$) then $G-U$ (or $G-u$) is the
graph obtained from $G$ by the removal of vertices of $U$ (vertex
$u$ correspondingly).

Given an edge  $e=\{x,y\}$ of a graph $G$, the graph  $G/e$ is
obtained from  $G$ by contracting $e$; which is, to get
$G/e$ we identify the vertices  $x$ and $y$ and remove all loops
and replace all multiple edges by simple edges. A graph $H$ obtained by a sequence of
edge-contractions is said to be a \emph{contraction} of $G$. A graph  $H$
is a \emph{minor} of $G$ if $H$ is a subgraph of a  contraction of
$G$.

We say that a graph $G$ is {\em $H$-minor-free} when it does not
contain $H$ as a minor. We also say that a graph class ${\cal G}$
is {\em $H$-minor-free} (or, excludes $H$ as a minor)  when
all its members are $H$-minor-free.

An \emph{apex graph} is a graph obtained from a planar graph $G$
by adding a vertex and making it adjacent to some of the vertices of $G$.
A graph class ${\cal G}$ is \emph{apex-minor-free} if ${\cal G}$
excludes a fixed apex graph $H$ as a minor.

The \emph{$(k\times k)$-grid} is the
Cartesian product of two paths of lengths $k-1$.

A \emph{surface} $\Sigma$ is a compact 2-manifold without boundary
(we always consider connected surfaces).
Whenever we refer to a {\em
$\Sigma$-embedded graph} $G$ we consider a  2-cell embedding of
$G$ in $\Sigma$. To simplify notations, we do not distinguish
between a vertex of $G$ and the point of $\Sigma$ used in the
drawing to represent the vertex or between an edge and the line
representing it.  We also consider a graph $G$ embedded in
$\Sigma$ as the union of the points corresponding to its vertices
and edges. That way, a subgraph $H$ of $G$ can be seen as a graph
$H$, where $H\subseteq G$.
Recall that $\Delta \subseteq \Sigma$ is
a (closed)  disc if it is homeomorphic to $\{(x,y):x^2 +y^2\leq1\}$.
The {\em Euler genus} of a nonorientable surface $\Sigma$
is equal to the nonorientable genus
$\tilde{g}(\Sigma)$ (or the crosscap number).
The {\em Euler genus}  of an orientable   surface
$\Sigma$ is $2{g}(\Sigma)$, where ${g}(\Sigma)$ is  the orientable genus
of $\Sigma$. We refer to the book of Mohar and Thomassen \cite{MoharT01} for
more details  on graphs embeddings.

If $X\subseteq 2^A$ for some set $A$, then by $\bigcup X$ we denote the union of all elements of $X$.

Recall that a \emph{hypergraph} $\cal H$ is a pair ${\cal
H}=(V({\cal H}), E({\cal H}))$ where $V({\cal H})$  is a finite
nonempty set of vertices, and $E({\cal H})$ is a set of nonempty
subsets of $V({\cal H})$ called \emph{hyperedges},
$\bigcup E(\H)=V(\H)$. We consider here
only hypergraphs without isolated vertices (i.e. every vertex is
in some hyperedge). 

For  vertex $v\in V(\H)$, we denote by
$E_{\H}(v)$  the set of its incident hyperedges.

The \emph{incidence graph} of the hypergraph $\cal H$  is the
bipartite graph $I({\cal H})$ with vertex set $V(\H)\cup E(\H)$
such that $v\in V(\H)$ and $e\in E(\H)$ are adjacent in $I(\H)$ if
and only if $v\in e$.

\subsection{Treewidth  of graphs and hypergraphs} 

A \emph{tree decomposition} of a hypergraph  ${\cal H}$ is a pair $(T,\chi)$,
where $T$ is a tree and $\chi\colon V(T)\to 2^{V({\cal H})}$ is a function
associating a set of vertices $\chi(t)\subseteq V({\cal H})$ (called a
\emph{bag}) to each node $t$ of the decomposition tree $T$ such that
i) $V({\cal H})=\bigcup_{t\in V(T)}\chi(t)$,
ii) for each $e\in
E({\cal H})$, there is a node $t\in V(T)$ such that $e\subseteq \chi(t)$, and
iii) for each $v\in V(G)$, the set $\{t\in V(T)\colon v\in\chi(t)\}$
forms a subtree of $T$.

The \emph{width} of a tree decomposition
equals $\max\{|\chi(t)|-1\colon t \in V(T)\}$.
The \emph{treewidth}
of a hypergraph $\H$ is the minimum width over all tree decompositions of
$\H$. We use notation $\tw(\H)$ for the treewidth of a hypergraph $\H$.

It is easy to verify that
for any hypergraph $\H$, $\tw(\H)+1\geq \tw(I(H))$. However, these parameters can differ
considerably on hypergraphs. For example, for the $n$-vertex hypergraph $\H$ with one
hyperedge which contains all vertices, $\tw(\H)=n-1$ and
$\tw(I(\H))=1$.

Since $\tw(\H)\geq |e|$ for every $e\in E(\H)$, we have that the presence of a
large hyperedge results in a large treewidth of the hypergraph. The
paradigm shift in the transition from treewidth to hypertree width
consists in counting the covering hyperedges rather than counting
the number of vertices in a bag. This parameter seems to be more
appropriate, especially with respect to constraint satisfaction
problems. We start with the introduction of even more general parameter of
fractional hypertree width.

\subsection{Hypertree width, its generalizations and related notions}
\label{ssecbas}

In general, given a set $A$, we use the term {\em labeling of $A$}
for any function $\gamma: A\rightarrow [0,1]$. We also use
the notation $\mathscr{G}(A)$ for the collection of
all labellings of a set ${\cal A}$.

The {\em size} of a labelling of $A$ is defined as
$|\gamma|=\sum_{x\in A}\gamma(x)$.
If the values of a labelling $\gamma$ are restricted to be $0$ or $1$, then
we say that $\gamma$ is a {\em binary} labelling of $A$. Clearly, the size of a binary
labelling is equal to the number of the elements of $A$ that are labelled by 1.
Given a  hyperedge labelling $\gamma$ of a hypergraph ${\cal H}$,
we define the set of vertices of ${\cal H}$ that are {\em blocked}
by $\gamma$  as  $$B(\gamma)=\{v\in V({\cal H})\mid \sum_{e\in E_{{\cal H}}(v)}\gamma(e)\geq 1\},$$
i.e. the set of vertices that are incident to hyperedges whose total labelling sums up to 1 or more.

A \emph{fractional hypertree decomposition}  \cite{GroheM06}  of
$\H$ is a triple $(T,\chi,\lambda)$, where $(T,\chi)$ is a tree
decomposition of $\H$ and
$\lambda\colon V(T)\to \mathscr{G}(E({\cal H}))$ is
a function, assigning a hyperedge labeling to each node of $T$,
such that for every $t\in V(T)$,
$\chi(t)\subseteq B(\lambda(t))$, i.e. all vertices of the bag $\chi(t)$
are blocked by the labelling $\lambda(t)$.
The \emph{width} of a fractional hypertree decomposition
$(T,\chi,\lambda)$ is $\min\{|\lambda(t)|\colon t\in V(T)\}$, and
the \emph{fractional hypertree width} $\fhw({\cal H})$ of $\cal H$
is the minimum of the widths of all fractional hypertree
decompositions of $\cal H$. 

If $\lambda$ assigns  a binary hyperedge labeling to each node of $T$, then $(T,\chi,\lambda)$ is a
\emph{generalized hypertree decomposition}  \cite{GottlobLS02}. Correspondingly,
the \emph{generalized hypertree width} $\ghw({\cal H})$ of $\cal H$
is the minimum of the widths of all generalized hypertree
decompositions of $\cal H$. 

Clearly, $\fhw({\cal H})\leq\ghw({\cal H})$ but, as it was shown in \cite{GroheM06},
 there are families of hypergraphs of bounded fractional hypertree width but unbounded generalized hypertree width.
Notice that computing the fractional hypertree width is an {\sf NP}-complete problem even for sparse graphs. To see this,  take a connected graph $G$ that is not a tree and construct a new graph $H$ by replacing every edge of G by $|V(G)|+1$ paths of length 2.  It is easy to check that $\tw(G)+1=\fhw(H)$.

The proof of the  next lemma follows from results of
\cite{ChekuriR00} about query width. For completeness, we
provide a direct proof here.

\begin{lemma}\label{lem:tw_bound}
For any hypergraph ${\cal H}$,  $ \fhw(\H)\leq\ghw(\H)\leq\tw(I({\cal H}))+1$. 
\end{lemma}

\begin{proof}
Let $(T,\chi)$ be a tree decomposition of $I(\H)$ of width
$\leq k$. It is enough to describe a generalized hypertree decomposition
$(T,\chi^\prime,\lambda)$ for $\H$ that has width $\leq k$ . For every $t\in V(T)$, let
$\chi^\prime(t)=(\chi(t)-E(\H))\cup (\bigcup(\chi(t)\cap
E(\H)))$. We include to $\lambda(t)$ all hyperedges $\chi(t)\cap
E(\H)$, and for every $v\in\chi(t)\cap V(\H)$, a hyperedge $e$ such
that $v\in e$ is chosen arbitrary and included to $\lambda(t)$.
Clearly, $V(\H)=\bigcup\limits_{t\in V(T)}\chi^\prime(t)$, for each
$e\in E(\H)$ there is a node $t\in V(T)$ such that
$e\subseteq\chi^\prime(t)$, and for every $t\in V(T)$
$\chi^\prime(t)\subseteq\bigcup\lambda(t)$. We have to prove that
for each $v\in V(\H)$, the set $\{t\in V(T)\colon
v\in\chi^\prime(t)\}$ forms a subtree of $T$. Suppose that there are
$s,t\in V(T)$ at distance at least two,
$v\in\chi^\prime(s)\cap\chi^\prime(t)$ and
$v\notin\chi^\prime(x)$ for all inner vertices $x$ of $s,t$-path in
$T$. Since $(T,\chi)$ is a tree decomposition of $I(\H)$,
$s\in\chi^\prime(t)-\chi(t)$ or $t\in\chi^\prime(s)-\chi(s)$. Assume
that $t\in\chi^\prime(t)-\chi(t)$. It means that there is
$e\in\chi(t)$ such that $v\in e$. Note that $e\notin\chi(x)$ for
inner vertices $x$ of $s,t$-path (otherwise $v\in\chi^\prime(x)$ by
the definition). If $v\in\chi(s)$ then there is no bag in $(T,\chi)$
that contains both endpoints of the edge $\{v,e\}\in E(I(\H))$. So
$s\in\chi^\prime(s)-\chi(s)$ and there is $e^\prime\in\chi(s)$ such
that $v\in e^\prime$. As before $e^\prime\notin\chi(x)$ for inner vertices and $e\neq e^\prime$. But since $v$ is adjacent with $e$
and $e^\prime$ in $I(\H)$, bags $\chi(x)$ contain $v$
and we receive a contradiction.
\end{proof}

It is necessary to remark here that the  fractional hypertree width
of a hypergraph can be arbitrarily smaller that the treewidth of its
incidence graph. Suppose that a hypergraph $\H^\prime$ is obtained
from the hypergraph $\H$ by adding a hyperedge which includes all
vertices. Then $\fhw(\H^\prime)=1$ and
$\tw(I(\H^\prime))+1\geq\tw(I(\H))+1\geq\fhw(\H)$.

Let ${\cal H}$ be a hypergraph. Two sets $X,Y\subseteq V(\H)$  \emph{touch} if $X\cap Y\neq\emptyset$ or
there exists $e\in E({\cal H})$ such that $e\cap X\neq\emptyset$ and
$e\cap Y\neq\emptyset$. A \emph{hyperbramble} of $\H$ is a set ${\cal B}$
of pairwise touching connected subsets of $V(\H)$ \cite{AdlerGG07}.
We say that a labelling $\gamma$ of $E({\cal H})$ {\em covers} a vertex
set $S\subseteq V({\cal H})$ if some of its vertices are blocked by $\gamma$.
The
{\em fractional order} of a hyperbramble is the minimum $k$ for which there
is a labeling 
$\gamma$ of size at most $k$ covering all elements in ${\cal B}$.
The
\emph{fractional hyperbramble number}, $\fbn(\H)$,
of $\H$ is the maximum of
the fractional orders of all hyperbrambles of $\H$.

The {\em  robber and army game}
was introduced by Grohe and Marx in  \cite{GroheM06}.
 The game is played on
a hypergraph ${\cal H}$  by two players, the robber and
the general who commands the army. A position of the game is a pair $(\gamma, v)$, where
$\gamma$ is a labelling of $E({\cal H})$ and
$v \in  V({\cal H})$. The choice of $\gamma$ is a distribution of the army
on the hyperedges
of ${\cal H}$,
chosen by the general, while $v$ is the position of the robber.
During the game, a vertex of the hypergraph is only blocked if the
total amount of army on the hyperedges that contain this vertex
adds up to the strength of at least one battalion.
To start a play of the game,
the robber picks a position  $v_{0}$, and the initial position is
$(\mathscr{O}, v_{0})$, where $\mathscr{O}$ denote the constant
zero mapping.
In each round, the players move from the current position
$(\gamma, v)$ to a new position $(\gamma, v')$ as follows: The general
selects $\gamma'$, and then the robber selects $v'$ such that there is
a path from $v$ to $v'$ in the hypergraph ${\cal H}$
that avoids the vertices in  $B(\gamma) \cap B(\gamma')$.
Under these circumstances,  the positions $(\gamma,v)$ and $(\gamma',v')$ are called
{\em compatible}. A {\em game sequence} is a sequence of compatible positions
and its cost is the maximum size of a distribution $\gamma$ in it.
If, at some moment, the position of the game is $(\gamma,v)$
where $v\in B(\gamma)$, then the general wins. If this never happens,
then the robber wins. A {\em winning strategy of cost at most $k$}
for the general is a program that provides a response
on each possible position such that any game sequence
generated by this program 
is finite and has cost at most $k$.
The {\em army width}, $\aw({\cal H})$,
of ${\cal H}$ is the least $k$ for which there exist
a winning strategy of cost at most $k$.

Using the fact that $\aw({\cal H}) \leq  \fhw({\cal H})$ (\cite[Theorem~11]{GroheM06}),  we can prove the following lemma.

\begin{lemma}
\label{lem:lowerb}
For any hypergraph ${\cal H}$, $\fbn({\cal H})\leq \fhw({\cal H})$.  
\end{lemma}

\begin{proof}
Let ${\cal B}$ be a hyperbramble of ${\cal H}$ of fractional order at least $k$.
Our aim is to provide an escape strategy for the robber against
any possible winning strategy of cost at most $<k$. In particular,
the robber will always be on a vertex of some set $S\in {\cal B}$
such that $S$ not covered by $\gamma$ and at any position $(\gamma,v)$
of the game there will be a new unblocked vertex for the robber to move.
Indeed, if the response of the general at position $(\gamma,v)$
is $\gamma'$, we have that $|\gamma|<k$ and therefore $\gamma$
cannot cover all elements of ${\cal B}$. If $S'\in {\cal B}$ is such a set,
the new position of the robber will be any vertex $v'$ of $S'$. Clearly, the robber can move
from $v$ to $v'$, as $S$ and $S'$
touch and all of their vertices are unblocked. This implies that
$\fbn({\cal H})\leq \aw(\H)$ and the result follows from the fact that $\aw({\cal H}) \leq  \fhw({\cal H})$,
proved in~\cite[Theorem~11]{GroheM06}.
\end{proof}

The variant of the robber and army game where the labellings
are restricted to be binary labellings is called the  {\sl Marshals and
Robbers game} and was introduced by Gottlob et al. \cite{GottlobLS03}.
The corresponding parameter is called {\em Marshall width} and is denoted as $\mw$.
Clearly, for any hypergraph $\H$, $\aw(\H)\leq \mw(G)$.

\subsection{i-brambles}

An \emph{$i$-labeled graph} $G$ is a triple $(G,N,M)$ where $N,M\subseteq V(G)$,
$N\cup M=V(G)$, $M-N$ and $N-M$ are independent sets of $G$, and for any $v\in V(G)$
its closed neighborhood  $N_{G}[v]$ is intersecting both $N$ and $M$. Notice that
$\{N,M\}$ is not necessarily a partition of $V(G)$.
The  incidence graph $I({\cal H})$
of a hypergraph ${\cal H}$ can be seen as an $i$-labeled
graph $(I({\cal H}),N,M)$  where $N=V({\cal H})$, $M=E({\cal H})$.

The result of the contraction of an edge  $e=\{x,y\}$ of an $i$-labeled
graph $(G,N,M)$ to a vertex $v_{e}$ is the  $i$-labeled graph $(G',N',M')$
where i) $G'=G\slash e$ ii)  $N'$ contains all
vertices of $N-\{x,y\}$ and also the vertex $v_{e}$, in case
$\{x,y\}\cap N\neq \emptyset$ and iii) $M'$ contains all
vertices of $M-\{x,y\}$ and also the vertex $v_{e}$, in case
$\{x,y\}\cap M\neq \emptyset$. An $i$-labeled graph $(G',N',M')$ is a {\em contraction}
of an $i$-labeled graph $(G,N,M)$ if $(G',N',M')$ can be obtained after applying a (possibly empty)
sequence of contractions to $(G,N,M)$. The following lemma is a direct  consequence of the definitions.

\begin{lemma}
\label{lem:contcon}
Let $(G,N,M)$ be an $i$-labeled graph and let $G'$ be a contraction of $G$.
Then there are $N',M'\subseteq V(G')$ such that the  $i$-labeled graph $(G',N',M')$ is a contraction of $(G,N,M)$.
\end{lemma}

Let $(G,N,M)$ be an $i$-labeled graph.  We say that a set $S\subseteq N$
is  {\em $i$-connected}  if any pair $x,y\in S$ is connected by a path in $G[S\cup M]$.
We say that two subsets  $S,R\subseteq N$ $i$-{\em touch} either if i) $S\cap R\neq \emptyset$, or
ii) there is an edge $\{x,y\}$ with $x\in S$ and $y\in R$,
or iii) there is a vertex $z\in M$ such that $N_{G}[z]$ intersects both $S$ and $R$.

Given an $i$-labeled graph $(G,N,M)$ we define an $i$-bramble of $(G,N,M)$ as
any collection ${\cal B}$ of $i$-touching $i$-connected sets of vertices in $N$.
We say that a labeling $\gamma$ of $M$ {\em controls} a vertex $x\in N$ if  $\sum_{y\in N_{G}[x]\cap M}\gamma(y)\geq 1$.
We say that $\gamma$ {\em fractionally covers} a vertex set $S\subseteq N$
if some of its vertices is controlled by $\gamma$.
The {\em order} of an $i$-bramble is the minimum $k$ for which there
is a labeling $\gamma$ of $M$ of size  at most $k$  that fractionally
covers all sets of  ${\cal B}$.

The {\em fractional $i$-bramble} number $\ifb(G,N,M)$ of an $i$-labeled graph $(G,N,M)$
is the maximum order of all $i$-brambles of it.

The following statement follows immediately from the definitions of hyperbrambles and $i$-brambles.

\begin{lemma}\label{lem:equiv}
For any hypergraph ${\cal H}$, $\ifb(I({\cal H}),V({\cal H}),E({\cal H}))=\fbn({\cal H})$.
\end{lemma}

Also it can be easily seen that the fractional $i$-bramble number is a contraction-closed parameter.

\begin{lemma}\label{lem:contr}
If an $i$-labeled graph $(G',N',M')$ is the contraction of an $i$-labeled graph $(G,N,M)$
then $\ifb(G',N',M')\leq \ifb(G,N,M)$.
\end{lemma}

Obviously, $i$-bramble number is not a subgraph-closed parameter (not even for induced subgraphs), but we can note the following useful claim.

\begin{lemma}
\label{lem:sub}
Let $(G,N,M)$ be an $i$-labeled graph and $X\subseteq V(G)$ such that $G-X$ has no isolated vertices, and for every $v\in X\cap M$, $N_G[v]\subseteq X$. Then $(G-X,N-X,M-X)$ is an $i$-labeled graph and $\ifb(G-X,N-X,M-X)\leq\ifb(G,N,M)$.  
\end{lemma}

\begin{proof}
Let $G^\prime=G-X$, $N^\prime=N-X$ and $M^\prime=M-X$.
Since $G^\prime$ has no isolated vertices, $(G^\prime,N^\prime,M^\prime)$ is an $i$-labeled graph.
Let $\cal B$ be an $i$-bramble of $(G^\prime,N^\prime,M^\prime)$. Obviously, $\cal B$ is an $i$-bramble of $(G,N,M)$, and there is a labeling $\gamma$ of $M$ of size $k\leq\ifb(G,N,M)$ which fractionally covers all sets of $\cal B$. It is now enough to note only that the restriction
$\gamma^\prime$ of $\gamma$ to $M$ is the labeling of $M^\prime$ which covers all sets of $\cal B$ and $|\gamma^\prime|\leq k$.
\end{proof}

\section{When hypertree width is sandwiched by treewidth}

\subsection{Influence and valency of $i$-brambles}

Let $(G,N,M)$ be an $i$-labelled graph and ${\cal B}$ an $i$-bramble of it.
We define the {\em influence} of ${\cal B}$, as
  $\ifl({\cal B})= \max_{v\in \cup{\cal B}}|\{x\in \cup{\cal B}\mid {\bf dist}_{G}(v,x)\leq 2\}|.$
We also
define the  {\em valency} of ${\cal B}$ as the quantity
$\val({\cal B})=\max_{v\in \cup{\cal B}}|\{S\in {\cal B}\mid v\in S\}|.$

\begin{lemma}
\label{lem:main-br}
If ${\cal B}$ is an $i$-bramble of an $i$-labeled graph $(G,N,M)$, then the  order of ${\cal B}$ is at least $\frac{|{\cal B}|}{\ifl({\cal B})\cdot \val({\cal B})}$.  
\end{lemma}

\begin{proof}
Let $\gamma$ be a labelling of $M$ that fractionally covers all sets of ${\cal B}$. We first prove the following claim.

\noindent{\em Claim.} $\gamma$ controls at most $\ifl({\cal B})\cdot |\gamma|$ vertices in $N({\cal  B})$.

\noindent{\em proof.}  Let $R$
be a subset of $\cup{\cal B}$ such that every vertex in $R$ is controlled by $\gamma$.
We define  $G_{R}$ as the graph whose vertex set is $R$
and where two vertices $x,y\in R$
are adjacent if their distance in $G$ is 1 or 2. By the definition of influence,
we obtain that the maximum degree of $G_{R}$
is at most $\ifl({\cal B})-1$ and therefore, $G_{R}$ has
an independent set $I$ of size at least $|R|/\ifl({\cal B})$.
As $I\subseteq R$, all vertices of $I$ are controlled by $\gamma$. This implies that
$\forall {x\in I} \sum_{y\in N_{G}[x]\cap M}\gamma(x)\geq 1$.
By definition, for each pair $x,x'\in I$, $x\neq x^\prime$, $N_{G}[x]\cap N_{G}[x']=\emptyset$.
Therefore,
$$|\gamma|=\sum_{x\in M}\gamma(x)\geq \sum_{x\in N_G[R]\cap M}\gamma(x)\geq \sum_{x\in N_G[I]\cap M}\gamma(x)\geq
\sum_{x\in I}\sum_{y\in N[x]\cap M}\gamma(y)\geq |I|\geq \frac{|R|}{\ifl({\cal B})},$$
and the claim follows.

The above claim, along with the definition of valency, implies that $\gamma$  fractionally covers no more than  $\ifl({\cal G})\cdot |\gamma|\cdot \val({\cal B})$ sets
of ${\cal B}$. We conclude that $|{\cal B}|\leq \ifl({\cal G})\cdot |\gamma|\cdot \val({\cal B})$ and the lemma follows.
\end{proof}

\subsection{Triangulated grids} 

A \emph{partially triangulated $(k\times k)$-grid} is
a graph $G$ that is obtained from a  $(k\times k)$-grid (we refer to it as its {\em underlying grid})
after adding some edges
without harming the planarity of the resulting graph.
Each vertex of $G$ will be denoted by a pair $(i,j)$ corresponding
to its coordinates in the underlying grid.  We will also denote as $U(G)$ the vertices, we call them \emph{non-marginal},
of $G$ that in the underlying grid have degree 4 and we call the vertices in $V(G)-U(G)$
{\em marginal}.

\begin{lemma}
\label{lem:fractriangrid}
 Let $(G,N,M)$ be an $i$-labeled graph,  where $G$ is a
partially triangulated $(k\times k)$-grid for $k\geq 4$. Then   $\ifb(G,N,M)\geq k/50-c$, for some constant $c\geq 0$.
\end{lemma}

\begin{proof}
We use notation  $C_{i,j}$ for the set vertices of $N\cap U(G)$ that belong to the $i$-th row or
the $j$-th column of the underlying grid of $G$. We claim that
${\cal B}=\{C_{i,j}\mid 2\leq i,j\leq k-1\}$
is an $i$-bramble of $G$ of order $\geq k/50-c$, for some constant $c\geq 0$.
Since $k\geq 4$, we have that each set $C_{i,j}$ is non-empty and $i$-connected.
Notice also that the intersection of the $i$-th row and the $j^\prime$-th column of the underlying grid of $G$
is either a vertex in $N$ and $C_{i,j}\cap C_{i^\prime,j^\prime}\neq\emptyset$, or
a vertex in $M-N$, but then all neighbors of it in $G$ belong to $N$. Therefore, all $C_{i,j}$ and $C_{i^\prime,j^\prime}$
should $i$-touch, and ${\cal B}$ is an $i$-bramble.
Each vertex $v=(i,j)$ in $N({\cal B})$ is contained in exactly $2k-5$ sets
of ${\cal B}$ (that is $k-2$ sets $C_{i',j'}$ that agree on the first coordinate plus $k-2$ sets  $C_{i',j'}$ that agree on the second, minus one set 
$C_{i,j}$ that agrees on both), therefore $\val({\cal B})=2k-5$. For each non-marginal vertex $x$ in $G$,
there are at most $25$ non-marginal vertices within distance $\leq 2$ in $G$ (in the worst case,  consider a triangulated $(5\times 5)$-grid
subgraph of $G$ that is centered at $x$) and thus $\ifl({\cal B})\leq 25$.
As  $|{\cal B}|=(k-2)^{2}$, Lemma~\ref{lem:main-br} implies that there is a constant $c$ such that the order of ${\cal B}$ is at least
$k/50-c$ and the lemma follows.
\end{proof}

\begin{theorem}\label{thm:planar}
If $\H$ is a hypergraph with a planar incidence graph $I(H)$, then
$\fhw(\H)-1\leq\ghw(\H)-1\leq \tw(I(\H))\leq 300 \cdot \fhw(\H)+c$ for some constant $c\geq 0$.
\end{theorem}

\begin{proof}
The left hand inequality follows directly from Lemma~\ref{lem:tw_bound}.
Suppose now that $\H$ is a hypergraph where  $\fhw(\H)\leq  k$.
By Lemmata~\ref{lem:lowerb} and~\ref{lem:equiv}, $\ifb(I({\cal H}),V({\cal H}),E({\cal H}))=\fbn({\cal H})\leq \fhw({\cal H})\leq k$.
By Lemmata~\ref{lem:contr} and~\ref{lem:fractriangrid}, $(I({\cal H}),V({\cal H}),E({\cal H}))$ cannot be $i$-contracted to an $i$-labeled graph $(G,N,M)$
where $G$ is a partially triangulated $(l\times l)$-grid, where  $l=50\cdot k+O(1)$.
By Lemma~\ref{lem:contcon}, ${\cal I}({\cal H})$ cannot be contracted to a partially triangulated $(l\times l)$-grid and thus
$I({\cal H})$ excludes an $(l\times l)$-grid as a minor.
  From
\cite[(6.2)]{RobertsonST94}, $\tw(I({\cal H}))\leq 6\cdot l\leq 300\cdot k+ c$ and the result follows.
\end{proof}

\subsection{Brambles in Gridoids}

We call a graph $G$ by a {\em $(k,g)$-gridoid} if it is possible to obtain a partially triangulated
$(k\times k)$-grid after
removing at most $g$ edges from it (we call these edges {\em additional}).

\begin{lemma}
\label{lem:moregen}
Let $(G,N,M)$ be an $i$-labeled graph where $G$ is a
$(k,g)$-gridoid. Then   $\ifb(G,N,M)\geq k/50-c\cdot g$ for some constant $c\geq 0$.  
\end{lemma}

\begin{proof}
The proof goes the same way as the proof of Lemma~\ref{lem:fractriangrid}. The only
difference is that now we exclude from ${\cal B}$ all the $C_{i,j}$'s where
either $i$ or $j$ is the coordinate of some endpoint of an additional edge.
Notice that again $\val({\cal B})\leq 2k-5$. Moreover, it also holds $\ifl({\cal B})\leq 25$ as none of the endpoints
is  in $N({\cal B})$ or $M({\cal B})$. Finally $|{\cal B}|\geq (k-2-2\cdot g)^{2}$ and the result follows
from Lemma~\ref{lem:main-br}.
\end{proof}

The proof of the next theorem is similar to the one of Theorem~\ref{thm:planar} (use Lemma~\ref{lem:moregen} instead of Lemma~\ref{lem:fractriangrid} and  \cite[Theorem~4.12]{DemaineFHT05jacm} instead of~\cite[(6.2)]{RobertsonST94}.

\begin{theorem}
\label{thm:genusb}
If $\H$ is a hypergraph with an incidence graph $I(H)$ of Euler genus at most  $g$, then
$\fhw(\H)-1\leq\ghw(\H)-1\leq \tw(I(\H))\leq 300\cdot g \cdot \fhw(\H)+c\cdot g$, for some constant $c\geq 0$.  \end{theorem}

\subsection{Brambles in augmented grids}

An  {\em augmented $(r \times r)$-grid of
span $s$} is an {\em $r \times r$ grid} with some extra edges such
that each vertex of the resulting graph is attached to at most $s$ non-marginal vertices
of the grid.

\begin{lemma}
\label{lem:brbr}
 If $(G,N,M)$ is an $i$-labeled graph where $G$ is an
augmented $(k\times k)$-grid with span $s$, then   $\ifb(G,N,M)\geq \frac{k}{2\cdot s^{2}}-c$, for some constant $c\geq 0$.    
\end{lemma}

\begin{proof}
We consider the $i$-bramble ${\cal B}=\{C_{i,j}\mid 2\leq i,j\leq k-1\}$ of the proof of Lemma~\ref{lem:fractriangrid}
and we directly observe that $\val({\cal B})\leq 2k-5$ and $|{\cal B}|\geq (k-2)^{2}$. By the definition
of the augmented $(k\times k)$-grid with span $h$ we obtain that  $\ifl({\cal B})\leq s^{2}$ and the result follows
applying Lemma~\ref{lem:main-br}.
\end{proof}

As it was shown by  Demaine et al. \cite{DemaineFH05}, every apex-minor-free graph with
treewidth at least $k$ can be contracted to a $(f(k)\times f(k))$-augmented grid of span $O(1)$ (the hidden constants in
the ``$O$''-notation depend
only on the excluded apex). Because,  $f(k)=\Omega(k)$
(due to the results of Demaine and Hajiaghayi in~\cite{DemaineH05II}), we have the following proposition.

\begin{proposition}
\label{prop:augmb}
 Let $G$ be an $H$-apex-minor-free graph of  treewidth at least $ c_{H}\cdot k$.  Then $G$  contains as a contraction
an augmented $(k\times k)$-grid of span $s_{H}$, where  constants
$c_{H}, s_{H}$ depend only on the size of apex graph $H$ that is excluded.
\end{proposition}

The proof of the next theorem is similar to the one of Theorem~\ref{thm:planar} (use Lemma~\ref{lem:brbr} instead of Lemma~\ref{lem:fractriangrid} and Proposition~\ref{prop:augmb} instead of~\cite[(6.2)]{RobertsonST94}.

\begin{theorem}\label{thm:apex-m-free}
If $\H$ is a hypergraph with an incidence graph $I(\H)$ that is  $H$-apex-minor-free, then
$\fhw(\H)-1\leq\ghw(\H)-1\leq \tw(I(\H))\leq  c_{H}\cdot \fhw(\H)$ for some constant $c_{H}$ that depends only on $H$.
\end{theorem}

\section{Hypergraphs with $H$-minor-free incidence graphs}
The results of Theorem~\ref{thm:apex-m-free} cannot be extended to
 hypergraphs which incidence graph
excludes an arbitrary fixed graph $H$ as a minor. For example, for every integer
$k$, it is possible to construct a hypergraph $\H$ with the planar incidence graph such that $\tw(I(\H))\geq k$.
By adding to $\H$ an universal hyperedge
containing all vertices of $\H$, we obtain a hypergraph $\H'$ of generalized hypertree width
one. Its incidence graph $I(\H')$ does not contain the complete graph $K_6$
as a minor, however its treewidth is at least $k$.
Despite of that, in this section we prove that if a  hypergraph
has  $H$-minor-free incidence graph, then its generalized hypertree width and
fractional hypertree width can be approximated by  the treewidth of
a graph that can be constructed from its incidence graph in polynomial time.
By making use of this result we show that in this case generalized
 hypertree width and
fractional hypertree width are up to a constant multiplicative factor from
each other. Another consequence of the combinatorial result is that there is a
constant factor polynomial time approximation algorithm for both parameters on this class
of hypergraphs.
 Our proof is based on the Excluded Minor Theorem by
 Robertson and Seymour \cite{RobertsonS03}.

\subsection{Graph minor theorem}

Before describing the Excluded Minor Theorem we need some
definitions.

\begin{definition}[{\sc Clique-Sums}]
Let $G_1=(V_1,E_1)$ and $G_2=(V_2,E_2)$ be two disjoint graphs, and
$k \geq 0$ an integer. For $i=1,2$, let $W_i \subseteq V_i$, form a
clique of size $h$ and let $G'_i$ be the graph obtained from $G_i$
by removing a set of edges (possibly empty) from the clique
$G_i[W_i]$. Let $F: W_1 \rightarrow W_2$ be a bijection between
$W_1$ and $W_2$. We define the $h$-\emph{clique-sum}
of $G_1$ and $G_2$, denoted by $G_1
\oplus_{h,F} G_2$, or simply $G_1 \oplus G_2$ if there is no
confusion, as the graph obtained by taking the union of $G_1'$ and
$G_2'$ by identifying $w\in W_1$ with $F(w) \in W_2$, and by
removing all the multiple edges. The image of the vertices of $W_1$
and $W_2$ in $G_1 \oplus G_2$ is called the \emph{join} of the sum.
\end{definition}

Note that some edges of $G_1$ and $G_2$ are not edges of $G$, since it is possible that they were added by clique-sum operation. Such edges are called \emph{virtual} edges of $G$.
We remark that $\oplus$ is not well defined; different choices of
$G'_i$ and the bijection $F$ could give different clique-sums. A
sequence of $h$-clique-sums, not necessarily unique, which result in
a graph $G$, is called a \emph{ clique-sum decomposition} of $G$.

\begin{definition}[$h$-nearly embeddable graphs]
Let $\Sigma$ be a surface with boundary cycles $C_1, \dots,C_h$,
i.e. each cycle $C_i$ is the border of a disc in $\Sigma$. A graph
$G$ is {\em $h$-nearly embeddable} in $\Sigma$, if $G$ has a subset
$X$ of size at most $h$, called {\em apices}, such that there are
(possibly empty) subgraphs $G_0,\dots,G_h$ of $G-X$ such
that
i) $G - X=G_0\cup\dots\cup G_h$,
ii) $G_0$ is embeddable in $\Sigma $, we fix an embedding of $G_0$,
iii) graphs $G_1,\dots,G_h$ (called \emph{vortices}) are pairwise disjoint,
iv) for $1\leq\dots\leq h$, let $U_i:= \{u_{i_1},\dots,u_{i_{m_i}}\} = V(G_0) \cap V(G_i)$,  $G_i$ has a path decomposition $(B_{ij}),\ 1\leq j \leq m_i$, of width at most $h$ such that
a) for $1\leq  i \leq h$ and for $1 \leq j \leq m_i$ we have $u_j \in B_{ij}$,
b) for $1\leq i \leq h$, we have $V(G_0) \cap C_i = \{u_{i_1},\dots,u_{i_{m_i}}\} $ and the points $u_{i_1},\dots,u_{i_{m_i}}$ appear on $C_i$ in this order (either if we walk clockwise or anti-clockwise).
\end{definition}

The following proposition is known as the Excluded Minor Theorem  \cite{RobertsonS03}
and is the cornerstone  of
Robertson and Seymour's Graph Minors theory.
\begin{theorem}[\cite{RobertsonS03}]
\label{structurethm}
For every non-planar graph $H$, there exists an integer $h$, depending only on
the size of $H$, such that every
graph excluding $H$ as a minor can be obtained by $h$-clique-sums from graphs that
can be $h$-nearly embedded in a surface $\Sigma$ in which $H$ cannot be embedded.
\end{theorem}
Let us remark that by the result of
Demaine et al.~\cite{DemaineHK05} such a
clique-sum decomposition can be obtained in time $O(n^c)$
for some constant $c$ .
 which depends only from $H$ (see also \cite{DawarGK07}).

\subsection{Approximation}

Let $\H$ be a hypergraph such that its incidence graph $G=I(\H)$
excludes a fixed graph $H$ as a minor. 
Every graph excluding a planar graph $H$ as a minor has a constant treewidth \cite{RobertsonST94}. Thus if $H$ is planar, 
by the results of the previous section, the generalized hypertree width does not exceed some constant. In what follows,  we always assume that $H$ is not planar.

By Theorem~\ref{structurethm},
there is an $h$-clique-sum decomposition of
 $G=G_1\oplus G_2\oplus\dots\oplus G_m$ such that
  for every $i\in\{1,2,\dots,m\}$, the summand $G_i$ can be $h$-nearly
  embedded in a surface $\Sigma$ in which $H$ can not be embedded.
  We assume that this clique-sum decomposition is \emph{minimal}, in the
  sense that
   for every virtual edge $\{x,y\}\in E(G_i)$ there is an $x,y$-path in
    $G$ with all inner  vertices in $V(G)-V(G_i)$.
Let $A_i$ be the set of apices of $G_i$. We define
$E_i=A_i\cap E(\H)$ and
 $G_i^\prime=G_i-(N_G[E_i]\cup A_i)$. 
For every virtual edge $\{x,y\}$ of $G_i^\prime$ we perform the
 following operation: if there is no $x,y$-path in $G-(N[E_i]\cup A_i)$
with all inner vertices in $G-V(G_i^\prime)$,
 then $\{x,y\}$ is removed from  $G_i^\prime$.
We denote the resulted  graph  by $F_i$.

In what remains we show that the maximal value of $\tw(F_i)$, where maximum is taken over
all  $i\in\{1,2,\dots,m\}$, is a constant factor approximation of
generalized and fractional hypertree widths of $\H$.
 The upper bound is given by the following lemma. Its proof
 uses the fact that $\ghw(\H)\leq
3\cdot\mw(\H)+1$ (see \cite{AdlerGG07}) and is based on the description
of a  winning strategy for
$k=\max\{\tw(F_i)\colon i\in\{1,2,\dots,m\}\}+2h+1$ marshals on $\H$.

\begin{lemma}\label{lem:upper}
$\ghw(\H)\leq 3\cdot \max\{\tw(F_i)\colon i\in\{1,2,\dots,m\}\}+6h+4$.  
\end{lemma}

\begin{proof}
Let $w=\max\{\tw(F_i)\colon i\in\{1,2,\dots,m\}\}$ and $k=w+2h+1$.
By the result of Adler et al. \cite{AdlerGG07},
we have that $\ghw(\H)\leq
3\cdot\mw(\H)+1$, and it is enough to describe a  winning strategy for
$k$ marshals on $\H$.

The clique-sum decomposition $G=G_1\oplus G_2\oplus\dots\oplus
G_m$ can be considered as a tree decomposition $(T,\chi)$ of $G$
for some tree $T$ with nodes $\{1,2,\dots,m\}$ such that
$\chi(i)=V(G_i)$, i.e. the  vertex send of the summands are the bags
of this decomposition. The idea behind the winning  strategy for
marshals
is to ``chase''
the robber in the hypergraph along $m+1$ decompositions for its incidence graph: one is induced by the
clique-sum decomposition and others are tree decompositions of
$F_i$. We say that marshals \emph{block} a set $X\subseteq V(G)$
if all hyperedges $X\cap E(\H)$ are occupied by them, and for
every $v\in X\cap V(\H)$, there is an occupied by a marshal
hyperedge $e\in E(\H)$ such
that $v\in e$.

Let us note that the definition  of $F_i$ yields the following:
if $x,y\in V(F_i)$, and
there is a $x,y$-path in $G-(N[E_i]\cup A_i)$ with all inner
vertices not in $F_i$, then $\{x,y\}$ is an edge of $F_i$.
(Indeed, if $\{x,y\}$ is an edge of $G$, then it is also an edge of
$F_i$.
If $\{x,y\}\notin
E(G)$ but such a path exits, then $\{x,y\}$ is a virtual edge in $G_i$
and by the definition of $F_i$, such an edge also is an edge of $F_i$.)

For $i\in\{1,2,\dots,m\}$,
let $(T^{(i)},\chi_i)$
be  a tree decomposition of $F_i$ of width at most $w$.
We  assume that trees $T$ and $T^{(1)},T^{(2)},\dots,T^{(m)}$ are
 rooted trees with roots $r$ and $r_1,r_2,\dots,r_m$ correspondingly.

For a node  $i\in V(T)$ and its parent  $j$ (in $T$), we define
 $S=V(G_i)\cap V(G_j)$. (If $i=r$ then we put $S=\emptyset$.)
  By the definition of the clique-sum,  $|S|\leq h$. Assume that
at most  $h$ marshals are already placed on the hypergraph in such
a way that they block $S$. Assume also that the robber occupies
some vertex of $\chi(T_i)$. We put
at most $h$ marshals on hyperedges to block the set of apices $A_i$.
Then the set $N_G[E_i]\cup A_i$ is
 blocked by these marshals.

Now marshals start to ``chase'' the robber
in the subhypergraph induced by the vertex set $V(F_i)\cap V(\H)$
along $T^{(i)}$.
We put at most $w+1$ marshals to block the set $\chi_i(r_i)$.
Assume now that some set $\chi_i(x)$ for $x\in V(T^{(i)})$ is blocked,
and that the robber can only occupy vertices of $\chi_i(T^{(i)}_y)$, where
$T^{(i)}_y$ is a subtree of $T^{(i)}$ rooted in
some child $y$ of the node $x$. We remove some marshals which were
placed to block $\chi_i(x)$ in such a way that $\chi_i(x)\cap\chi_i(y)$
 remains blocked, and then place additional marshals to block
 $\chi_i(y)$. This manoeuvre can be done by making use of
at most $w+1$ marshals. We put $x=y$ and repeat this operation
 until there is a child $y$ of $x$ such that the robber can be in
 $\chi_i(T^{(i)}_y)$. Thus by repeating  at most $|V(T^{(i)})|$ times this
 operation, marshals  ``push'' the robber out of $V(F_i)\cap V(\H)$.

Let $j$ be a child of $i$ in $T$ such that the robber now can occupy
 only the vertices of $\chi(T_j)$, where $T_j$ is the subtree of $T$ rooted at $j$.
 Let $S^\prime=V(G_i)\cap V(G_s)$.
Since $|S^\prime|\leq h$, we have that
$h$ marshals can block this set and, after that, all
other marshals can be removed from $\H$.

We apply the described strategy of marshals starting from $i=r$
until the robber is captured  in  some leaf-node of $T$.
For every node of $T$ we have used at most $h$ marshals to occupy
apices, at most $h$ marshals to block the vertices of the clique-sum, and
at most $w+1$ marshals to push the robber out of $F_i$. Thus in total
at most $2h+w+1$ marshals have a winning strategy on $\H$.
\end{proof}

To prove the lower bound we  need the following property of the
 clique-sum decomposition which was observed by Demaine and Hajiaghayi
\cite{DemaineH05II}.

\begin{proposition}\label{prop:clique_sum_restr}
Each clique sum in the expression $G=G_1\oplus G_2\oplus\dots\oplus G_m$ involves at most three vertices from each summand other than apices and vertices in vortices of that summand.
\end{proposition}

We also need a result roughly stating that if a
graph $G$ with a big grid as a surface minor is embedded
on a surface $\Sigma$ of small genus, then
there is a disc in $\Sigma$ containing a big part of the grid of $G$.
This result is implicit in the work of Robertson and Seymour and there are
simpler alternative proofs by Mohar and Thomassen \cite{Mohar92,Thomassen97} (see also \cite[Lemma 3.3]{DemaineFHT05jacm}). We use the following variant
of this result from Geelen et al. \cite{GeelenRS04}.

\begin{proposition}[\cite{GeelenRS04}] \label{prop_Geelen}
Let $g, l,r$ be positive integers such that $r \geq g(l + 1)$ and let $G$ be an $(r, r)$-grid. If $G$ is embedded in a surface $\Sigma$ of Euler genus at most $g^2 - 1$, then some
$(l,l)$-subgrid of $G$ is embedded in a closed disc $\Delta$ in $\Sigma$
such that the boundary cycle of the $(l,l)$-grid is the boundary of the disc.
\end{proposition}

Now we are ready to prove the following  lower bound.

\begin{lemma}\label{lem:lower}
$\fbn(\H)\geq \varepsilon_H
\cdot\max\{\tw(F_i)\colon i\in\{1,2,\dots,m\}\}$
 for some constant $\varepsilon_H$ depending only on $H$.
\end{lemma}

\begin{proof}
We assume that $G-(N[E_i]\cup A_i)$ is a connected graph which has
at least one edge. (Otherwise one can consider the
components of this graph separately and remove isolated vertices.)
The main idea of the proof is to contract it to a planar graph with approximately
the same treewidth as $F_i$ and then apply same techniques that were used in the
previous section for the planar case.

\paragraph{Structure of $G-(N[E_i]\cup A_i)$.}
Let us note that  an $h$-clique-sum decomposition $G=G_1\oplus
G_2\oplus\dots\oplus G_m$ induces an $h$-clique-sum decomposition
of $G^\prime=G-(N[E_i]\cup A_i)$ with the  summand $G_i$ replaced by
$F_i$.
 Let
$G^\prime_1,G^\prime_2,\dots,G^\prime_l$ be the connected components of
$G^\prime-V(F_i)$. Every such component $G^\prime_j$ is attached
via clique-sum to $F_i$ by  some clique $Q_j$ of $F_i$.
Note that cliques $Q_j$ contain all
virtual edges of $F_i$. We assume that each
clique $Q_j$ does not separate vertices of $F_i$.
Otherwise, it is possible to decompose
$F_i$ into the clique-sum of  graphs $F^{(1)}_i\oplus F^{(2)}_i$
with the join $Q_j$ and prove the bound for summands and, since
$\tw(F_i)=\max\{F^{(1)}_i,F^{(2)}_i\}$, that will prove the lemma. To simplify the structure of the graph,
for every component $G^\prime_j$, we contract all its edges
and denote by $S_j$ the star whose central vertex is the result of the
contraction and leaves are the vertices of $Q_j$.

\paragraph{Contracting vortices.}
The $h$-nearly embedding of the graph $G_i$ induces the $h$-nearly
embedding of $F_i=X_0\cup X_1\cup\dots\cup X_h$ without apices. Here we
assume that  $X_0$ is embedded in a surface $\Sigma$ of genus depending on $H$
and $X_1,X_2,\dots,X_h$ are the
vortices. For every vortex $X_j$, the vertices $V(X_0)\cap V(X_j)$ are on the boundary $C_j$ of some face of $X_0$. If for a star $S_k$ some of its leaves
 $Q_k$ are in  $X_j$ or $C_j$, we do the
following operation:
if $Q_k\cap (V(X_j)-V(C_j))\neq\emptyset$
then all edges of $S_k$ are contracted, and if $Q_k\cap
(V(X_j)-V(C_j))=\emptyset$ but $|Q_k\cap V(C_j)|\geq 2$,  then we contract all
edges of $S_k$ that are incident to the vertices of $Q_k\cap V(C_j)$. These contractions results in the contraction of some
edges of $F_i$. Particularly, all virtual edges of $X_j$ and $C_j$ are contracted.  Additionally, we contract all remaining edges of
$X_j$ and $C_j$. We perform theses contractions  for all vortices of $F_i$ and denote the result by $F^\prime_i$. It follows
immediately from the definition of the $h$-clique-sum and
Proposition~\ref{prop:clique_sum_restr}, that $F^\prime_i$
coincides with the graph obtained from $F_i$ by contractions of
all vortices $X_j$ and boundaries of faces $C_j$.
It can be easily
seen that $F^\prime_i$ is embedded in $\Sigma$. It is known (see
e.g. \cite{DemaineFHT05jacm,DemaineH05II}) that there is a
positive constant $a_H$ which depends only on $H$ such that
$\tw(F^\prime_i)\geq a_H\cdot\tw(F_i)$.

\paragraph{Contracting the part that lies outside of some planar disc.}
Since $F^\prime_i$ is embedded in $\Sigma$, we have that the graph $F^\prime_i$  contains some $(k\times k)$-grid as a surface minor, where $k\geq b_H\cdot\tw(F^\prime_i)$ for some constant $b_H$ \cite{DemaineFHT05jacm}. Combining this result with Proposition~\ref{prop_Geelen}, we receive the following claim. There is
a disc $\Delta\subseteq\Sigma$ such that
i) the subgraph $R$ of $F^\prime_i$ induced by vertices of $F^\prime_i\cap \Delta$ is a connected graph;
 ii) the subgraph $R^\prime$ of $F^\prime_i$ induced by $N_{F^\prime_i}[V(R)]$ is completely in some disc $\Delta^\prime$;
iii) vertices of $V(R^\prime)-V(R)$ induce a cycle $C$ which is the border of $\Delta^\prime$,
and
iv) $\tw(R)\geq c_H\cdot\tw(F^\prime_i)$ for some constant $c_H$.
Now we treat the part of $F^\prime_i$ which is outside
$\Delta$ exactly the same way we have treated vortices.
For stars $S_k$ intersecting  $V(F^\prime_i)-V(R^\prime)$ or $C$, we do the following:
if $Q_k\cap (V(F^\prime_i)-V(R^\prime))\neq\emptyset$, then all edges of $S_k$ are contracted, and if $Q_k\cap (V(F^\prime_i)-V(R^\prime))=\emptyset$ but
$|Q_k\cap V(C)|\geq 2$, then all edges of $S_k$ incident to the vertices of $Q_k\cap V(C)$ are contracted. These contractions result in the contraction of some edges of $F^\prime_i$ with endpoints  on $C$ or outside $\Delta^\prime$. Particularly, all such virtual edges are contracted.  Additionally, we contract all remaining edges of $F^\prime_i-V(R)$ and $C$.
Thus this part of the graph is contracted to a single vertex.
Denote the obtained graph $X$. This graph is planar, and since $R$ is a subgraph of $X$, we have that $\tw(X)\geq\tw(R)$.

\paragraph{Embedding the stars.}
Some edges of $X$ are virtual, and all such edges are  in cliques $Q_j$.
By Proposition~\ref{prop:clique_sum_restr},   $|Q_j|\leq3$.
For every clique $Q=V(X)\cap Q_j$, we do the following. If $Q=\{x,y\}$, then the edge of the star $S_j$ incident to $x$ is contracted. If $Q=\{x,y,z\}$, then if two vertices of $Q$, say $x$ and $y$, are joined by an edge in $G$, then the edge of $S_j$ incident to $z$ is contracted, and if there are no such edges and the
triangle induced by $\{x,y,z\}$ is the boundary of some face of $X$, then we add a new vertex on this face, join it with $x$, $y$ and $z$ (it can be seen as  $S_j$ embedded in this face, and since our graph is $i$-labeled, it is assumed that this new vertex has same labels as the central vertex of $S_j$), and then remove virtual edges.
Note that if the triangle is not a boundary of some face, then $Q$ is a separator of our graph, but we assumed that there are no such separators. Denote by $Y$ the obtained graph.
Similar construction was used in the proof of the main theorem in \cite{DemaineH05II}, and
by the same arguments as were used by Demaine et al.
we immediately conclude that there is a positive constant $d_H$ such that $\tw(X)\geq d_H\cdot \tw(Y)$.

Now all contractions are finished.
Note that the graph $Y$ is a planar graph which is
a contraction of $G^\prime=G-(N[E_i]\cup A_i)$. Also there is some positive constant $e_H$ which depends only on $H$ such that $\tw(Y)\geq e_H\cdot\tw(F_i)$.
Recall that we consider the $i$-labeled graph $(G,V(\H),E(\H))$. By Lemma~\ref{lem:equiv},
$\fbn(\H)=\ifb(G,V(\H),E(\H))$. Because the sets $V(\H)$ and $E(\H)$ are independent, by Lemma~\ref{lem:sub},  we have that $\ifb(G,V(\H),E(\H))\geq\ifb(G^\prime,N,M)$, where $N=V(\H)- (N[E_i]\cup A_i)$ and $M=E(\H)-(N[E_i]\cup A_i)$.
By Lemma~\ref{lem:contr}, $\ifb(G^\prime,N,M)\geq\ifb(Y,N^\prime,M^\prime)$, where
$N^\prime$ and $M^\prime$ are sets which were obtained as the result of contractions of $N$ and $M$. Finally,
as in Theorem~\ref{thm:planar}, one can show  that $\ifb(Y,N^\prime,M^\prime)\geq f_H\cdot \tw(Y)$ for some constant $f_H$.
By putting  all these bounds together, we prove that
there is a positive constant $\varepsilon_H$ which depends only on $H$,  such that $\fbn(\H)\geq \varepsilon_H\cdot\tw(F_i)$.
\end{proof}

Combining Lemmata~\ref{lem:tw_bound}, \ref{lem:lowerb},  \ref{lem:upper}, and \ref{lem:lower},
we obtain the following theorem.

\begin{theorem}\label{thm:appr}
$(1/c_H)\cdot w \leq \fhw(\H)\leq \ghw(\H)\leq c_H \cdot w$, where
$
w=\max\{\tw(F_i)\colon i\in\{1,2,\dots,m\}\}$, and $c_H$ is a constant depending
 only on $H$.
\end{theorem}

\noindent{\bf Remark.}
Notice that, by Theorem~\ref{thm:appr},
the generalized hypertree width and the fractional hypertree width
of a  hypergraph with $H$-minor-free incidence graph may differ within a  multiplicative
constant factor.  We stress that, as observed in  \cite{GroheM06},
this is not the case for general hypergraphs.

Demaine et al. \cite{DemaineHK05} (see also \cite{DawarGK07,RobertsonS03})
described an algorithm which constructs
a clique-sum decomposition of an $H$-minor-free graph $G$ on $n$ vertices
with the
running time $n^{O(1)}$ (the hidden constant in the running time depends
only on $H$).
As far as we constructed summands $G_i$, the construction of graphs $F_i$
can be done in polynomial time. 
Moreover, since the algorithm  of Demaine et al. provides
$h$-nearly embeddings
of these graphs, it is possible to use it to construct a polynomial
 constant factor approximation
algorithm for the computation of $\tw(F_i)$.
This provides  us with the main algorithmic result of this section.

\begin{theorem}\label{thm:alg}
For any fixed graph $H$, there is a polynomial time
$c_H$-approximation algorithm computing  the generalized
hypertree width and the fractional hypertree width
for hypergraphs with $H$-minor-free incidence graphs, where the constant
$c_H$ depends only on $H$.
\end{theorem}

Let us remark  that while the winning
strategy for marshals used in the proof of Lemma~\ref{lem:upper} is not
monotone (a strategy is \emph{monotone} if the territory available
for the robber only decreases in the game), but it can be turned into
 monotone by choosing  marshals' positions in
a slightly more careful way. By making use of the results from
\cite{GottlobLS03}, the monotone
strategy  can be used  to construct a generalized
hypertree decomposition (or fractional hypertree decomposition).
Thus our results can be used not only to compute but to construct,
up to constant multiplicative-factor, the corresponding decompositions.



\end{document}